\documentclass[11pt]{article}
\usepackage{amsmath}
\usepackage{fancyhdr}
\usepackage{amssymb}
\usepackage{amsthm}
\usepackage{graphicx}
\usepackage{varioref}
\newtheorem{theorem}{Theorem}
\newtheorem{definition}{Definition}
\newtheorem{lemma}{Lemma}
\usepackage{verbatim} 
\usepackage{multicol}
\usepackage{enumerate}
\usepackage[normalem]{ulem}
\usepackage{subfig}
\usepackage[T1]{fontenc}
\usepackage[margin=1in]{geometry}
\usepackage{fancyhdr}

\usepackage{url}
\usepackage{xspace}
\usepackage{thm-restate}

\usepackage{hyperref}
\hypersetup{
    bookmarksnumbered=true, 
    unicode=false, 
    pdfstartview={FitH}, 
    pdftitle={Quantum Adversary (Upper) Bound}, 
    pdfauthor={Shelby Kimmel}, 
    pdfsubject={}, 
    pdfcreator={}, 
    pdfproducer={}, 
    pdfkeywords={}, 
    pdfnewwindow=true, 
    colorlinks=true, 
    linkcolor=blue, 
    citecolor=blue, 
    filecolor=blue, 
    urlcolor=blue 
}

\newcommand{\eq}[1]{\hyperref[eq:#1]{(\ref*{eq:#1})}}

\renewcommand{\sec}[1]{\hyperref[sec:#1]{Section~\ref*{sec:#1}}}
\newcommand{\thm}[1]{\hyperref[thm:#1]{Theorem~\ref*{thm:#1}}}
\newcommand{\lem}[1]{\hyperref[lem:#1]{Lemma~\ref*{lem:#1}}}
\newcommand{\prop}[1]{\hyperref[prop:#1]{Proposition~\ref*{prop:#1}}}
\newcommand{\cor}[1]{\hyperref[cor:#1]{Corollary~\ref*{cor:#1}}}
\newcommand{\fig}[1]{\hyperref[fig:#1]{Figure~\ref*{fig:#1}}}
\newcommand{\app}[1]{\hyperref[app:#1]{Appendix~\ref*{app:#1}}}
\newcommand{\defi}[1]{\hyperref[defi:#1]{Definition~\ref*{defi:#1}}}
\newcommand{\clm}[1]{\hyperref[clm:#1]{Claim~\ref*{clm:#1}}}

\begin{document}

\title{Quantum Adversary (Upper) Bound\footnote{Conference Version Appeared in ICALP 2012: Automata, Languages, and Programming
Lecture Notes in Computer Science Volume 7391, 2012, pp 557-568}}
\author{Shelby Kimmel\\
Center for Theoretical Physics, \\
Massachusetts Institute of Technology\\
\texttt{skimmel@mit.edu}}


\date{}

\maketitle
\begin{abstract}
We describe a method to upper bound the quantum query complexity of  Boolean formula evaluation problems, using fundamental theorems about the general adversary bound. This nonconstructive method can give an upper bound on query complexity without producing an algorithm. For example, we describe an oracle problem which we prove (non-constructively) can be solved in $O(1)$ queries, where the previous best quantum algorithm uses a polylogarithmic number of queries. We then give an explicit $O(1)$-query algorithm for this problem based on span programs.
\end{abstract}

\section{Introduction} \label{sec:intro}
The general adversary bound has proven to be a powerful concept in quantum computing. Originally formulated as a lower bound on the quantum query complexity of Boolean functions \cite{Hoyer2007}, it was proven to be a tight bound 
both for the query complexity of evaluating discrete finite functions and for 
the query complexity of the 
more general problem
of state conversion \cite{Randothers2}.
The general adversary bound is the culmination of a series of adversary methods \cite{Ambainis2000,Ambainis2003}. 
While the adversary method in its various forms has been useful in finding lower bounds on quantum query complexity \cite{Ambainis2010,Hoyer2008,Reichardtorigin}, the general adversary bound itself can be difficult to apply, as the quantity for even simple, few-bit functions must usually be calculated numerically \cite{Hoyer2007,Reichardtorigin}. 

One of the nicest properties of the general adversary bound is that it behaves well under composition \cite{Randothers2}. This fact has been used to lower bound the query complexity of evaluating composed total functions, and to create optimal algorithms for composed total functions \cite{Reichardtorigin}. Here, we extend one of the composition results to partial Boolean
functions, and use it to upper bound the query complexity of Boolean functions. We do this by obtaining an {\it{upper}} bound on the general adversary bound.

Generally, finding an upper bound on
the general adversary bound is just as difficult as finding an algorithm, as they are dual problems \cite{Randothers2}. However, using the composition property of the general adversary bound, when given an algorithm for a Boolean function $f$ composed $d$ times, we obtain an upper bound on the general adversary bound of $f$. Due to the tightness of the general adversary bound and query complexity, this  procedure
gives an upper bound on the query complexity 
of $f$, but because it is nonconstructive, it doesn't give any hint as to what the corresponding algorithm for $f$ might look like. The procedure is a bit counter-intuitive: we obtain information about an algorithm for a simpler function by creating an algorithm for a more complex function. This is similar in spirit to the tensor power trick, where an inequality between two terms is proven by considering tensor powers of those terms\footnote{See Terence Tao's blog, {\it{What's New}} ``Tricks Wiki article: The tensor power trick," http://terrytao.wordpress.com/2008/08/25/tricks-wiki-article-the-tensor-product-trick/}.   

We describe a class of oracle problems called \textsc{Constant-Fault Direct Trees} (introduced by Zhan et al. \cite{us}), for which this method proves the existence of an $O(1)$ query algorithm. While this method does not give an explicit algorithm, we show that a span program algorithm achieves this bound. The previous best
algorithm for \textsc{Constant-Fault Direct Trees} has a query complexity that is polylogarithmic in the size of the problem. 

In \sec{proofsec} we describe the upper bound on the general adversary bound.
In \sec{example} we apply this bound to \textsc{Constant-Fault Direct Trees}
 and prove the 
existence of a constant query algorithm. In \sec{spanpsec} we describe the span program based quantum
algorithm for \textsc{Constant-Fault Direct Trees}.

\section{A Nonconstructive Upper Bound on Query Complexity} \label{sec:proofsec}

Our procedure for creating a nonconstructive upper bound on query complexity relies on the fact that the general adversary bound behaves well under composition and is a tight lower bound on quantum query complexity. The standard definition of the general adversary bound is not necessary 
for our purposes, but can be found in \cite{Hoyer2008}, and an alternate
definition appears in \app{comp}.

Our procedure applies to Boolean functions. A function $f$ is Boolean if $f:S\rightarrow\{0,1\}$ with 
$S\subseteq\{0,1\}^n$.
Given a Boolean function $f$ and a natural number $d$,
we define $f^d$, ``$f$ composed $d$ times,'' recursively as 
$f^d=f\circ(f^{d-1},\dots,f^{d-1})$, where $f^1=f$. 

Now we state the main result:
\begin{theorem}
\label{thm:maintheorem}
Suppose we have a (possibly partial) Boolean function $f$ that is composed $d$ times, $f^d$, and a quantum algorithm for $f^d$ that requires $O(J^d)$ queries. Then $Q(f)=O(J)$, where $Q(f)$ is the bounded-error quantum query complexity of $f$. 
\end{theorem}

\noindent (For background on bounded-error quantum query complexity and quantum algorithms, see \cite{Ambainis2000}.) There are seemingly similar results in the literature; for example, Reichardt proves in \cite{Reichardtrefl} that the query complexity of a function composed $d$ times, when raised to the $1/d^{th}$ power, is equal to the adversary bound of the function, in the limit that $d$ goes to infinity. This result gives insight into the exact query complexity of a function, and its relation to the general adversary bound. In 
contrast, our result is a {\it{tool}} for upper bounding query complexity, possibly without gaining any knowledge of the exact query complexity of the function.

One might think that \thm{maintheorem} is useless because an algorithm for $f^d$ usually comes from
composing an algorithm for $f$. If $J$ is the query complexity of the algorithm for $f$, one expects the query complexity of the
resulting algorithm for $f^d$ to be at least $J^d$. In this case, \thm{maintheorem} gives no new insight. Luckily for us, composed quantum
algorithms do not always follow this scaling. If there is a quantum algorithm for $f$ that uses  $J$ queries, where $J$ is not optimal (i.e. is larger than the true bounded error quantum query complexity of $f$), then the number of queries used when the algorithm is composed $d$ times can be much less than $J^d$. If this is the case, and if the non-optimal algorithm for $f$ is the best known,  \thm{maintheorem} promises the existence of an algorithm for $f$ that uses fewer queries than the best known algorithm, but, as \thm{maintheorem} is nonconstructive, it gives no information as to the form of the algorithm. 

We need two lemmas to prove \thm{maintheorem}:
\begin{restatable}{lemma}{lemonee}
\label{lem:lemone}
For any Boolean function $f:S\rightarrow\{0,1\}$ with 
$S\subseteq\{0,1\}^n$ and natural number $d$,
\begin{equation} \label{eq:advcomp}
{\rm{ADV}}^{\pm}(f^d)\geq ({\rm{ADV}}^{\pm}(f))^d.
\end{equation}
\end{restatable}
\noindent  H\o yer et al. \cite{Hoyer2007} prove \lem{lemone} for total Boolean functions\footnote{While the statement of Theorem 11 in \cite{Hoyer2007} seems to apply to partial functions, it is mis-stated; their proof actually assumes total functions.}, and the result is extended to more general
total functions in \cite{Randothers2}. Our contribution is to extend the result in \cite{Randothers2} to partial Boolean functions. While \thm{maintheorem} still holds for total functions, the example we consider later in the paper involves partial functions. The proof of \lem{lemone} closely follows the proof in \cite{Randothers2} and
can be found in \app{comp}. 
\begin{lemma} \emph{(Lee, et al. \cite{Randothers2})}
For any function $f:S\rightarrow E$, with $S\in D^n$, and $E, D$ finite sets, the bounded-error quantum query complexity of $f$, $Q(f)$, satisfies
\begin{equation}
Q(f)=\Theta({\rm{ADV}}^{\pm}(f)). 
\end{equation}
\label{lem:lem2}
\end{lemma}

We now prove \thm{maintheorem}:
\begin{proof}
Given an algorithm for $f^d$ that requires $O(J^d)$ queries, by \lem{lem2},
\begin{equation} \label{eq:eq1}
{\rm{ADV}}^{\pm}(f^{d})=O(J^d).
\end{equation}
Combining Eq. \eq{eq1} and \lem{lemone},
\begin{equation}
({\rm{ADV}}^{\pm}(f))^d=O(J^d).
\end{equation}
Raising both sides to the $1/d^{th}$ power,
\begin{equation}
{\rm{ADV}}^{\pm}(f)=O(J).
\end{equation}
We now have an upper bound on the general adversary bound of $f$. Finally, using \lem{lem2} again,
we obtain 
\begin{equation}
Q(f)=O(J).
\end{equation}
\end{proof}

\section{Example where the General Adversary Upper Bound is Useful} \label{sec:example}

In this section we describe a function, called the \textsc{1-Fault Nand Tree}, for which \thm{maintheorem} gives a better upper bound on query complexity than any previously known quantum algorithm. 
The \textsc{1-Fault Nand Tree} was proposed by Zhan et al. \cite{us} to obtain a super-polynomial speed-up for a partial Boolean formula, and 
is a specific type of \textsc{Constant-Fault Direct Tree}, which was mentioned in \sec{intro}.
We first define the \textsc{Nand Tree}, and then explain the allowed inputs to the  \textsc{1-Fault Nand Tree}.

The \textsc{Nand Tree} is a complete, binary tree of depth $d$,  where each node is assigned a bit value. The leaves are assigned arbitrary values, and any internal node $v$ is given the value \textsc{nand}$(val(v_1),val(v_2))$,
where $v_1$ and $v_2$ are $v$'s children, and $val(v_i)$ is the value of the node $v_i$. 

To evaluate the \textsc{Nand Tree},
one must find the value of the root given an oracle for the values of the leaves. (The \textsc{Nand Tree} is equivalent to solving \textsc{nand}$^d$, although the composition we will use for \thm{maintheorem} is not the composition of the \textsc{nand} function, but of the \textsc{Nand Tree} as a whole.) For arbitrary inputs, Farhi et al. showed that there exists an optimal quantum algorithm in the Hamiltonian model to solve the \textsc{Nand Tree} in $O(2^{0.5d})$ time \cite{FarhiNAND1}, and this was extended to a standard discrete algorithm with 
quantum query complexity $O(2^{0.5d})$ \cite{Childs2007,Reichardt2010}. Classically, the best algorithm requires $\Omega(2^{0.753d})$ queries \cite{Saks1986}. Here, we consider the \textsc{1-Fault Nand Tree}, which is a
\textsc{Nand Tree} with a promise that the inputs satisfy certain conditions.

\begin{definition}\emph{(\textsc{1-Fault Nand Tree} \cite{us})} \label{defi:faulttree}
Consider a \textsc{Nand Tree} of depth $d$ (as described above). Then to each node $v$, we assign an
integer $\kappa(v)$ such that:
\begin{itemize}
\item $\kappa(v)=0$ for leaf nodes.
\item Otherwise $v$ has children $v_1$ and $v_2$
\subitem If $val(v_1)=val(v_2),$ $\kappa(v)=\max_{i\in\{1,2\}} \kappa(v_i)$, 
\subitem If $val(v_1)\neq val(v_2)$, let $v_i$ be the node such that $val(v_i)=0$. Then $\kappa(v)=1+\kappa(v_i)$.
\end{itemize}
A tree satisfies the $1$-fault condition if $\kappa(v)\leq1$ for any node $v$ in the tree. 
\end{definition}
\noindent{\bf{Notation:}} When a node has one child with value 1 and one child with value 0 ($val(v_1)\neq val(v_2)$),  we call the node $v$ a {\bf{fault}}. (Since \textsc{nand}$(0,1)=$ \textsc{nand}$(1,0)=1$, fault nodes must have value $1$, although not all $1$-valued nodes are faults.) 

The $1$-fault condition is a limit on the amount and location of faults within the tree. In a \textsc{1-Fault Nand Tree}, if a path moving from a root to a leaf encounters any fault node and then passes through the $0$-valued child of the fault node, there can be no further fault nodes on the path. An example of a \textsc{1-Fault Nand Tree} is given in \fig{NANDtree}. 

The condition of the \textsc{1-Fault Nand Tree} may seem strange, 
but it has a nice interpretation when considering the correspondence between
\textsc{Nand Trees} and game trees\footnote{See Scott Aaronson's blog, \emph{Shtetl-Optimized},
 ``NAND now for something completely different," http://www.scottaaronson .com/blog/?p=207}. The \textsc{1-Fault Nand Tree}
 corresponds to a game in which, if both players
play optimally, there is at most one point in the sequence of play where a player's choice affects the outcome of the game. Furthermore, if a player makes
the wrong choice at the decision point, the game again becomes a 
single-decision game, where if both players play optimally for the rest of play,
 there is at most one 
point where a player's choice affects the outcome of the game.

\begin{figure}[!ht] 
\center\includegraphics[width=2.7in]{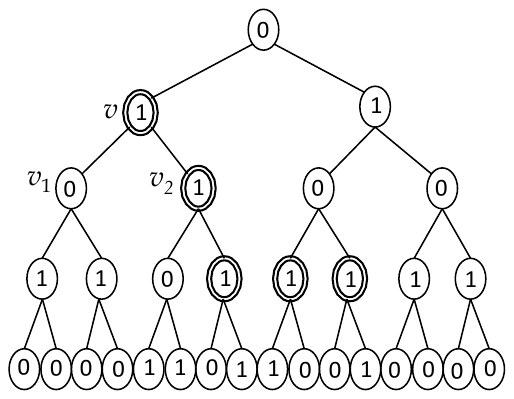}
 \caption{An example of a \textsc{1-Fault Nand Tree} of depth 4. Fault nodes are highlighted by a double circle. The node $v$ is a fault since one of its children ($v_1$) has value $0$, and one ($v_2$) has value $1$. Among $v_1$ and its children, there are no further faults, as required by the $1$-fault condition. There can be faults among $v_2$ and its children, and indeed, $v_2$ is a fault. There can be faults among the $1$-valued child of $v_2$ and its children, but there can be no faults below the $0$-valued child. }
 \label{fig:NANDtree}
  \end{figure}
Zhan et al. \cite{us} describe a quantum algorithm for the $d$-depth \textsc{1-Fault Nand Tree} that requires $O(d^2)$ queries to an oracle for the leaves. However, when the $d$-depth \textsc{1-Fault Nand Tree} is composed $\log d$ times, their
algorithm requires only $O(d^3)$ queries. Here we see an example where the number of queries required by a composed algorithm does not scale exponentially in the number of compositions, which is critical for applying \thm{maintheorem}. Applying \thm{maintheorem} to the algorithm for the \textsc{1-Fault Nand Tree} composed $\log d$ times, we find that an upper bound on the query complexity of the \textsc{1-Fault Nand Tree} is $O(1)$. This is a large improvement over $O(d^2)$ queries. Zhan et al. prove $\Omega(\rm{poly}\log d)$ is a lower bound on the classical query complexity of \textsc{1-Fault Nand Trees}. 
An identical argument can be used to show that \textsc{Constant-Fault Nand Trees} (from \defi{faulttree}, trees satisfying $\kappa(v)\leq c$ with $c$ a constant) have query complexity $O(1)$. 

In fact, Zhan et al. find algorithms for a broad range of trees, where instead of \textsc{nand}, the evaluation tree is composed of a type of Boolean function  called a {\it{direct}} function. A direct function is a generalization of a monotonic Boolean function, and includes functions like majority, threshold, and their negations. For the exact definition, which involves span programs, see \cite{us}. Similarly to \textsc{Constant-Fault Nand Trees}, Zhan et al. give a quantum algorithm for \textsc{Constant-Fault Direct Trees} requiring $O(d^2)$ queries and prove $\Omega(\rm{poly}\log d)$ is a lower bound on the classical query complexity, while
 \thm{maintheorem} can be used to prove the existence of $O(1)$-query algorithms for \textsc{Constant-Fault Direct Trees}.

\section{Span Program Algorithm for \textsc{Constant-Fault Direct Trees}}\label{sec:spanpsec}
The structure of \textsc{Constant-Fault Direct Trees} can be quite complex, and it is not obvious that there should be an $O(1)$-query algorithm. Inspired by the knowledge of the algorithm's existence, thanks to \thm{maintheorem}, we found a span program algorithm for \textsc{Constant-Fault Direct Trees} that requires $O(1)$ queries. It makes sense that the optimal algorithm uses span programs, not just because span programs can always
be used to create optimal algorithms \cite{Randothers2}, but because \thm{maintheorem}
is based on properties of the general adversary bound, and there is strong
duality between the general adversary bound and span programs.

Span programs are linear algebraic representations of Boolean functions, which have an intimate relationship with quantum algorithms. 
In particular, Reichardt proves \cite{Reichardtrefl} that given a span program $P$ for a function $f$, there is a function of the span
program called the witness size, such that one can create a quantum algorithm for 
$f$ with query complexity $Q(f)$ satisfying
\begin{equation}
Q(f)=O(\textsc{witness size}(P))
\label{eq:wsizeqq}
\end{equation}
Thus, creating a span program for a function is equivalent to creating a quantum query algorithm.

 There have been many iterations of span program-based quantum algorithms, 
due to Reichardt and others \cite{Randothers2,Reichardtrefl,Reichardtorigin}. Zhan et al. create algorithms for direct Boolean functions \cite{us} using the span program formulation described in Definition 2.1 in \cite{Reichardtrefl}, one of the earliest versions (we will not 
go into the details of span programs in this paper). Using the more recent advancements
in span program technology, we show here:
\begin{restatable}{theorem}{kfault}\label{thm:kfault}
Given an evaluation tree composed of the direct Boolean function $f$, with the promise that the tree satisfies the $k$-fault condition ($k$ a natural number), there is a quantum algorithm
that evaluates the tree using $O(w^k)$ queries, where $w$ is a constant that depends on $f$. In particular, for a \textsc{Constant-Fault Direct Tree} ($k$ a constant), the algorithm requires $O(1)$ queries.
\end{restatable}

While \thm{maintheorem} promises the existence of $O(1)$-query quantum algorithms for \textsc{Constant-Fault Direct Trees}, \thm{kfault} gives an explicit $O(1)$-query quantum algorithm for these problems. The proof combines properties of the witness size of direct Boolean functions with a more current version of span program algorithms. 

First we define a \textsc{$k$-Fault Direct Tree}, which is a Boolean evaluation
tree made up of a direct Boolean function composed many times, with a promise on the input. The definition
of direct Boolean functions are a bit technical and are given in \cite{us}; here we
will just use their properties. 
Given a direct Boolean function $f$ with $n$ inputs, a \textsc{Direct Tree} for $f$ is
a depth-$d$, $n$-partite complete graph, where each node is given a Boolean value.
The value of 
the node $v$, $val(v)$ is given by $val(v)=f(val(v_1),\dots,val(v_n))$ where $v_i$ is the $i^{th}$ child node of $v$. The values
of the leaves are given via an oracle, and the goal is to find the value of the root of the tree. A \textsc{$k$-Fault Direct Tree}
is a \textsc{Direct Tree} with inputs that satisfy certain conditions; the definition
is similar to \defi{faulttree} for \textsc{$1$-Fault Nand Trees}:
\begin{definition}
\label{defi:directfaulttree}
Let $T$ be a \textsc{Direct Tree} for $f$. Let each node be labeled as 
{\it{fault}} or {\it{trivial}} based on the values of its children and the specific function $f$ used. For
each node, depending on the values of its children, a set of its child nodes are labeled {\it{strong}} in relation to the node, and the remaining
child nodes are labeled {\it{weak}}. For trivial nodes, all children are strong. Then to each node $v$ we assign an
integer $\kappa(v)$ such that:
\begin{itemize}
\item $\kappa(v)=0$ for leaf nodes.
\item Otherwise $v$ has children $\{v_1,\dots,v_n\}$
\subitem If $v$ is trivial, then $\kappa(v)=\max_{i\in\{1,\dots,n\}} \kappa(v_i)=\max_{i:v_i \text{ is strong}}\kappa(v_i)$.
\subitem If $v$ is a fault, then $\kappa(v)=1+\max_{i:v_i \text{ is strong}}\kappa(v_i)$.
\end{itemize}
A tree satisfies the $k$-fault condition if $\kappa(v_r)\le k$ where $v_r$ is the root. \end{definition}

Notice that the restriction $\kappa(v_r)\le k$ 
is slightly relaxed compared to \defi{faulttree}, where it was required that $\kappa(v)\leq 1$ for all nodes $v$ in the whole tree. Thus the span 
program algorithm we describe below applies to an even broader range of trees than were included in 
the discussion in \sec{example}.

For $f=\textsc{nand}$, a node whose two children have the same value (either both $0$ or both $1$), is trivial, and a node with one $0$-valued child and one $1$-valued child is a fault. Furthermore, for $f=\textsc{nand}$, for fault nodes, the $0$-valued child is strong, and the $1$-valued child is weak. One can verify that with these designations, \defi{faulttree} corresponds to \defi{directfaulttree} for $f=\textsc{nand}$. 

For a function $f:S\rightarrow \{0,1\}$, with $S\subseteq\{0,1\}^n$, input $x\in S$, and span program $P$, the weighted witness
size on input $x$ is $\textsc{wsize}_s(P,x)$ where $s\in (\mathbb{R}^+)^n$ is the weighting vector. When $s=(1,\dots,1)$, we write $\textsc{wsize}_1(P,x)$. 
Following \cite{Reichardtrefl}, we rewrite Eq. \eq{wsizeqq} as 
\begin{equation}
Q(f)=O\left(\max_{x\in S}\text{ } \textsc{wsize}_1(P,x)\right) 
\label{eq:wsizeqf}
\end{equation}

In \cite{us}, Zhan et al. show that for any direct Boolean function $f$, one can create a span program $P$ with the following properties\footnote{We note 
that Zhan et al. use a different version of span programs than those used to prove Eq \eq{wsizeqf}. However Reichardt shows in \cite{Reichardtrefl} 
how to transform from one span program formulation to another, and proves that there is a transformation from the span program formulation used by Zhan et al.
to the one needed for \lem{wsizethm} that does not increase the witness size and that uses the weighting vector in the same way.}:
\begin{itemize}
\item $\textsc{wsize}_{1}(P,x)=1$ if input $x$ makes the function trivial.
\item $\textsc{wsize}_{1}(P,x)\leq w$, if input $x$ makes the function a fault, where $w$ is a constant depending only on $f$. 
\item For $\textsc{wsize}_s(P,x)$, $s_j$ do not affect the witness size, where the $j^{th}$ input bit is weak.
\end{itemize}

To create an algorithm, we will combine these facts with Eq. \eq{wsizeqf} and the following composition lemma:
\begin{lemma}(based on Theorem 4.3 in \cite{Reichardtrefl})
Let $f:S\rightarrow\{0,1\}$, $S\subseteq \{0,1\}^n$ and 
$g:C\rightarrow \{0,1\}$,
$C\subseteq \{0,1\}^m$, and consider the composed function $(f\circ g)(x)$ with
 $x=(x^1,\dots, x^n)$, $x^i\in C $ and $g(x^i)\in S$ $\forall i$. 
Let $\tilde{x}=(g(x^1),\dots, g(x^n))\in S$. Let $G$ be a
span program for $g$, $F$ be a span program for $f$, and $s\in(\mathbb{R}^+)^{n\times m}$. Then there
exists a span program $P$ for $f\circ g$ such that
\begin{equation}
\textsc{wsize}_s(P,x)\leq \textsc{wsize}_r(F,\tilde{x})\leq \textsc{wsize}_{1}(F,\tilde{x})\max_{i\in[n]}\textsc{wsize}_{s^i}(G,x^i)
\label{eq:wsizeeq}
\end{equation}
where $r=(\textsc{wsize}_{s^1}(G,x^1),\dots,\textsc{wsize}_{s^n}(G,x^n))$ and $s^i$ is a vector of the $i^{th}$ set of $m$ elements of $s$.
\label{lem:wsizethm}
\end{lemma}
\noindent The main
difference between this lemma and that in \cite{Reichardtrefl} is that the witness size here is input dependent, as is needed for partial functions. 
We also use a single inner function $g$ instead of $n$ different functions $g_i$, but we allow each inner function to have a different input. The proof of this result 
follows exactly using the proof of Theorem 4.3 in \cite{Reichardtrefl}, so we will not repeat it here.

Using the properties of strong and weak nodes in direct Boolean functions, we
see that $ \textsc{wsize}_r(F,\tilde{x})$ in \eq{wsizeeq} doesn't depend on $r_i=\textsc{wsize}_{s^i}(G,x^i)$ for weak inputs $i$. Thus,
we can rewrite Eq. \eq{wsizeeq} as
\begin{equation}
\textsc{wsize}_s(P,x)\leq\textsc{wsize}_{1}(F,\tilde{x})\max_{\substack{i\in[n]\\ \text{s.t. } i \rm{\text{ is strong}}}}\textsc{wsize}_{s^i}(G,x^i).
\label{eq:betterwsize}
\end{equation}

We now prove \thm{kfault}:
\begin{proof}
For a direct function $f$, we know there is a span program $P$ such that $\textsc{wsize}_{1}(P,x)\leq w$ for all fault inputs and $\textsc{wsize}_{1}(P,x)=1$ for trivial inputs. We will show that this implies the existence of a span program for the \textsc{k-Fault Direct Tree}
with  witness size $\leq w^k$. 

We will use an inductive proof on $k$, the number of faults. 
For the base case, consider a depth $1$ tree. This is just a single direct Boolean function. If its input makes the function a fault, using the properties of direct Boolean functions, there is
a span program for this input with witness size at most $w$. If the depth $1$ tree has has an input that makes the function trivial, there is
a span program for this input with witness size at most $1$.  Thus there
exists quantum algorithm with query complexity $O(1)$ that evaluates this tree.

Consider a depth $d$, $k$-fault tree $T$ with input $x$. We can think of this instead as a single direct function $f$ (with input $\tilde{x}$ and
span program $P$), composed with $n$ subtrees of depth $d-1$,
where we label the $i^{th}$ subtree $T^i$. Let $P_{T^i}$
be a span program for $T^i$, and we call the input to that subtree $x^i$. If $\tilde{x}$ makes $f$ a fault, then by 
 Eq \eq{betterwsize} we know there exists a span program $P_T$ for $T$ such that:
\begin{align}
\textsc{wsize}_1(P_T,x)&\leq\textsc{wsize}_1(P,\tilde{x})\times\max_{\substack{i\in[n]\\i: i\rm{\text{ is strong}}}}\textsc{wsize}_1(P_{T^i},x^i) \nonumber\\
&\leq w\times\max_{\substack{i\in[n]\\i: i\rm{\text{ is strong}}}}\textsc{wsize}_1(P_{T^i},x^i).
\end{align}
Now if we take the subtree $T^{i^*}$ that maximizes the $2^{nd}$ line, then by the definition of $k$-fault trees,
$T^{i^*}$ is a $(k-1)$-fault tree.
By inductive assumption, there is a span program for $T^{i^*}$ satisfying $\textsc{wsize}_1(P_{T^{i^*}}, x^i)\leq w^{k-1}$, so $T$ satisfies $\textsc{wsize}_1(P_T, x)\leq w^k$, and there is a quantum algorithm
for the tree that uses $O(w^k)$ queries.

Given the same setup, but now assuming the input $\tilde{x}$ makes $f$ trivial, then by Eq \eq{betterwsize} we have:
\begin{align}
\textsc{wsize}_1(P_T,x)&\leq\textsc{wsize}_1(P,\tilde{x})\times\max_{\substack{i\in[n]\\i: i\rm{\text{ is strong}}}}\textsc{wsize}_1(P^i,x^i) \nonumber\\
&=1\times\max_{\substack{i\in[n]\\i: i\rm{\text{ is strong}}}}\textsc{wsize}_1(P^i,x^i) .
\end{align}
Now if we take the subtree $T^{i^*}$ that maximizes the $2^{nd}$ line, then by the definition of fault trees,
$T^{i^*}$ is a $\kappa$ fault tree with $\kappa\leq k$. But we know if $\kappa\leq k-1$, then
$\textsc{wsize}_1(P_{T^{i^*}},x^i)\leq w^{k-1}$ by inductive assumption, so we're done in that case. So instead we assume $\kappa=k$. Thus
we have reduced the problem to a smaller depth tree, and we can repeat the above procedure until we find the first subtree with a 
fault at its root (in which case we are back to the previous case) or show that there are no further faults in the tree (in which case the tree can be evaluated in $O(1)$ queries). Since the tree has finite depth, this procedure will terminate.
\end{proof}

\section{Conclusions}
We describe a method for upper bounding the quantum query complexity of 
Boolean functions
using the general adversary bound. Using this method, we show that \textsc{Constant-Fault Direct Trees} can always be evaluated using $O(1)$ queries. Furthermore, we create an algorithm with a matching 
upper bound using span programs.

 We would like to find other examples where \thm{maintheorem} is useful, although we suspect that \textsc{Constant-Fault Direct Trees} are a somewhat unique case. It is clear from the span program algorithm described in \sec{spanpsec} that \thm{maintheorem} will not be useful for composed functions where the base function is created using this type of span program. However, there could be other types of quantum walk algorithms, for example, to which \thm{maintheorem} might be applied. In any case, this work suggests that new ways of upper bounding the general adversary bound could give us a second window into quantum query complexity
beyond algorithms.

Beside the practical application of \thm{maintheorem}, the result tells us something abstract and general about the structure of quantum algorithms. There is a natural way that quantum algorithms should compose, and if an algorithm does not compose in this natural way, then one knows that something is non-optimal.


\section{Acknowledgements} 
Many thanks to Rajat Mittal for generously explaining the details of the composition theorem for the general adversary bound. Thanks to the anonymous FOCS reviewer for pointing out problems with a previous version, and also for encouraging me to find a constant query span program algorithm. Thanks to Bohua Zhan, Avinatan Hassidim, Eddie Farhi, Andy Lutomirski, Paul Hess, and Scott Aaronson for helpful discussions. This work was supported by NSF Grant No. DGE-0801525,
{\em IGERT: Interdisciplinary Quantum Information Science and Engineering} and by the U.S. Department of Energy under cooperative research agreement Contract Number DE-FG02-05ER41360.

\appendix
\section{Composition Proof} \label{app:comp}
In this section, we will prove \lem{lemone}:
\lemonee*

\noindent This proof follows Appendix C from  Lee et al. \cite{Randothers2} very closely, including most notation. The difference
between this Lemma and that in \cite{Randothers2} is that $f$ is allowed to be partial. We write out most of the proof again because it is subtle where
 the partiality of $f$ enters the proof, and to allow this appendix to be read without constant reference to \cite{Randothers2}.

First, we use an expression for the general adversary bound derived from the dual program of the general adversary bound:
\begin{align}\label{eq:dual1}
{\rm{ADV}}^{\pm}(g)=&\max_{W,\text{ }\Omega\circ I =\Omega}W\bullet J \nonumber \\
&\text{subject to }W\circ G=0 \nonumber \\
&\hspace{1.8cm}\Omega\pm W \circ \Delta_i\succeq 0 \nonumber \\
&\hspace{1.8cm} \text{Tr}(\Omega)=1
\end{align}
\noindent where $g:C\rightarrow \{0,1\}$, with $C\subseteq\{0,1\}^m$ and all matrices are indexed by $x,y\in C$, so e.g. $[W]_{xy}$ is the element of $W$ in the row corresponding to input $x$ and column corresponding to input $y$. $W$ can always be chosen to be symmetric. $G$ satisfies $[G]_{xy}=\delta_{g(x),g(y)}$, and $\Delta_i$ satisfies $[\Delta_i]_{xy}=1-\delta_{x_i,y_i}$, with $x_i$ the value of the
$i^{th}$ bit of the input $x$. We call $\Delta_i$ the filtering matrix. $\succeq 0$ means positive semidefinite, $J$ is the all $1$'s matrix, and $W\bullet J \nonumber$ means take the sum of all elements of $W$. When $\circ$ is used between uppercase or
Greek letters, it denotes Hadamard product, while between lowercase letters, it denotes composition.

We want to determine the adversary bound for a composed function $f\circ g$ consisting of the functions $g:C\rightarrow\{0,1\}$ with 
$C\subseteq\{0,1\}^m$ and $f:S\rightarrow\{0,1\}$ with $S\subseteq\{0,1\}^n$. We consider the input to $f\circ g$ to be a vector of inputs $x=(x^1,\dots, x^n)$ with $x^i\in C$. Given
an input $x$ to the composed function, we denote the input to the $f$ part of the function as $\tilde{x}$: $\tilde{x}=(g(x^1),\dots,g(x^n))$. Let $(W, \Omega)$ be an optimal solution for $g$ with ${\rm{ADV}}^\pm(g)=d_g$ and $(V,\Lambda)$ be an optimal
solution for $f$ with ${\rm{ADV}}^\pm(f)=d_f$. To clarify the filtering matrices, we say $\Delta_q^g$ is indexed by inputs to $g$, $\Delta_p^f$ is indexed by inputs to $f$, and $\Delta_{(p,q)}^{f\circ g}$ is indexed by inputs
to the composed function $f\circ g$. (So $\Delta_{(p,q)}^{f\circ g}$ refers to the $(pm+q)^{\text{th}}$ bit of the input string.)

We assume that the initial input $x=(x^1,\dots,x^n)$ is valid for the $g$ part of the composition, i.e. $x^i\in C$ $\forall i$. A problem might
arise if $\tilde{x}$, the input to $f$, is not an element of $S$. This is an issue that Lee et al. do not have to deal with, but which might affect the proof. Here we show that the proof goes through with small modifications.

The main new element we introduce is a set of primed matrices, which extend the matrices indexed by inputs to $f$ to 
be indexed by all elements of $\{0,1\}^n$, not just those in $S$. For a primed matrix $A'$, indexed by $x, y\in\{0,1\}^n$, if $x\notin S$ or $y\notin S$, then $[A']_{xy}=0$.  
We use similar notation for matrices indexed by $x=(x^1,\dots,x^n)$ where $\tilde{x}\in S$; we create primed matrices by extending the indeces to all inputs $x$ 
by making those elements with $\tilde{x}\notin S$ have value $0$. Notice
if the extended matrices $(W',\Omega')$ are a solution to the dual program, then the reduced matrices $(W,\Omega)$
are also a solution.
For matrices $A'$  indexed by $\{0,1\}^n$, we define a new matrix $\tilde{A}'$ indexed by $C^n$, as $[\tilde{A}']_{xy}=[A']_{\tilde{x}\tilde{y}}$, where $\tilde{x}$ is the output of 
the $g$ functions on the input $x$, and likewise for $\tilde{y}$ and $y$. $\tilde{A}'$ expands each element of $A'$ into a block of elements.

Before we get to the main lemma, we will need a few other results:

\begin{lemma}\cite{Randothers2}
Let $M'$ be a matrix labeled by $x\in\{0,1\}^n$, and $\tilde{M}'$ be defined as above. Then if $M'\succeq 0$,
$\tilde{M}'\succeq 0$.
\label{lem:positive}
\end{lemma}
\begin{proof}
This claim is stated without proof in \cite{Randothers2}. $\tilde{M'}$ is created by turning all of the elements of $M'$ into block matrices with repeated inputs. When an index
$x\in \{0,1\}^n$ is expanded to a block of $k$ elements, there are $k-1$ eigenstates of $\tilde{M}'$ that
only have nonzero elements on this block and that have eigenvalue $0$. By considering all $2^n$ blocks (each element of $\{0,1\}^n$ becomes a block) we obtain $2^n(k-1)$ $0$-valued eigenvectors. Next we use the eigenvectors $\vec{v}^i$ of $M'$ to create new vectors $\vec{\tilde{v}}^i$ in the space of $\tilde{M}'$.
We give every element in the $x^{th}$ block of $\vec{\tilde{v}}^i$ the value $\vec{v}^i(x)/k$, where $\vec{v}^i(x)$ is the $x^{th}$ element of $\vec{v}^i$.
The vectors
$\vec{\tilde{v}}^i$ complete the basis with the $0$-valued eigenvectors, and are orthogonal to the $0$-valued vectors, but not 
to each other.  However, the $\vec{\tilde{v}}^i$ have the property that
$\vec{\tilde{v}}^{iT}\tilde{M}'\vec{\tilde{v}}^j=\delta_{ij}\lambda_i$ where $\lambda_i$ is the eigenvalue of $\vec{v}^i$, so $\lambda_i\geq 0$.
Thus using these vectors as a basis, we have that $\vec{u}^T\tilde{M}'\vec{u}\geq 0$ for all vectors $\vec{u}$.
\end{proof}

The following is identical to Claim C.1 from \cite{Randothers2} and follows because there is no restriction that
 $g$ be a total function. Thus we state it without proof:
\begin{lemma}
\label{lem:dualbool}
For a function $g$, there is a solution to the dual program, $(W,\Omega)$, such 
that ${\rm{ADV}}^\pm(g)=d_g$, $d_g\Omega\pm W\succeq 0$, and
 $\sum_{x:g(x)=1}\Omega(x,x)=\sum_{x:g(x)=0}\Omega(x,x)=1/2$.
\end{lemma}

In \lem{bigidea}, we will show that ${\rm{ADV}}^\pm(f\circ g)={\rm{ADV}}^\pm(f){\rm{ADV}}^\pm(g)$, which implies
\lem{lemone}.

\begin{lemma}\label{lem:bigidea}
A solution to the dual program for $f\circ g$ is $(U,\Upsilon)$, where $(U',\Upsilon')=(c\times\tilde{V}'\circ(d_g\Omega+W)^{\otimes n},c\times d^{n-1}_g\tilde{\Lambda}'\circ\Omega^{\otimes n})$ and $c=2^nd_g^{-(n-1)}$.  $(U,\Upsilon)$ give the adversary bound
${\rm{ADV}}^{\pm}(f\circ g)=d_gd_f$.
\end{lemma}
\begin{proof}
The first thing to check is that $U'$ and $\Upsilon'$ are valid primed matrices, or otherwise we can not recover $U$ and $\Upsilon$. Because each of $U'$ and $\Upsilon'$ are formed by Hadamard products with primed matrices, they themselves are also primed matrices. 

We next calculate the objective function, and afterwards check that $(U',\Upsilon')$ satisfy the conditions of the dual program. 

The objective function gives:
\begin{align}
J\bullet(c\tilde{V}'\circ(d_g\Omega+W)^{\otimes n})&=c\sum_{\substack{a,b\in S \\ f(a)\neq f(b)}}[V]_{ab}\sum_{\substack{x,y \\ \tilde{x}=a,\tilde{y}=b}}\prod_i(d_g[\Omega]_{x^iy^i}+[W]_{x^iy^i}) \nonumber \\
&=c \sum_{\substack{a,b\in S \\ f(a)\neq f(b)}}[V]_{ab}\prod_i\sum_{\substack{x^i,y^i \\ g(x^i)=a_i\\ g(y^i)=b_i }}(d_g[\Omega]_{x^iy^i}+[W]_{x^iy^i}) 
\label{eq:obj}
\end{align}
where in the first line we've replaced $V'$ by $V$ because adding extra $0$'s does not affect the sum. In the second line, $a_i$ and $b_i$ are the $i^{th}$ bits of $a$ and $b$ respectively, and we've changed the order of multiplication and addition. This ordering change is not affected by the fact that 
$f$ is partial, since the first summation already fixes an input to $f$. 

We now examine the sum 
\begin{align}
\sum_{\substack{x^i,y^i \\ g(x^i)=a_i\\ g(y^i)=b_i }}(d_g[\Omega]_{x^iy^i}+[W]_{x^iy^i}). 
\end{align}
We consider the cases $a_i=b_i$, and $a_i\neq b_i$ separately. When $a_i=b_i$, because $W\circ G=0$, we know that $[W]_{x^iy^i}=0$,
 so in this case, only $[\Omega]_{x^iy^i}$ is non-zero.
 Since $\Omega$ is diagonal, it only has non-zero values when $x^i=y^i$, and using \lem{dualbool}, the sum is 
$d_g/2$. When $a_i\neq b_i$, then $x^i\neq y^i$, so $[\Omega]_{x^iy^i}=0$. In this case, the sum will include exactly 
half of the elements of $W$: either those elements with $g(x^i)=0$ and $g(y^i)=1$, or with $g(x^i)=1$ and $g(y^i)=0$. Since $W$ is symmetric, this amounts to $\frac{1}{2}W\bullet J=d_g/2$. Multiplying $n$
times for the product over the $i's$ and using the definition of the objective function for $f$ gives the final result:
\begin{align}
J\bullet(c\tilde{V}'\circ(d_g\Omega+W)^{\otimes n})
=c\times d_f\left(\frac{d_g}{2}\right)^n=d_fd_g
\end{align}

Now we show that $U'$ and $\Upsilon'$ satisfy the conditions of the dual program. We require that $[U']_{xy}=0$
for $(f\circ g)(x)=(f\circ g)(y)$.  Notice $U'=0$ whenever $\tilde{V}'=0$,
and $[\tilde{V}']_{xy}=0$ for $(f\circ g)(x)=(f\circ g)(y)$, so this requirement
 holds. Likewise $\Upsilon'$ is a diagonal matrix because it can only be nonzero
where $\Omega^{\otimes n}$ is non-zero, and $\Omega^{\otimes n}$ is
 diagonal.

Next we will show that $\Upsilon'\pm U'\circ(\Delta^{f\circ q}_{(p,q)})'\succeq 0$. From
 \lem{dualbool} and from the conditions on the dual programs for $f$ and $g$,
 we have $d_g\Omega\pm W\succeq 0$, $\Omega\pm
 W\circ\Delta_q^g\succeq 0$, and 
$\Lambda'\pm V'\circ(\Delta_p^f)'\succeq 0$.
 Then by \lem{positive},
 $\tilde{\Lambda}'\pm\tilde{V}'\circ(\tilde{\Delta}_p^f)'\succeq 0$. Since tensor 
and Hadamard products preserve semidefinite positivity, we get
\begin{align} \label{eq:semipos}
0\preceq(\tilde{\Lambda}'\pm\tilde{V}'\circ(\tilde{\Delta}_p^f)')\circ\left((d_g\Omega+W)^{\otimes (p-1)}\otimes(\Omega+W\circ\Delta_q^g)\otimes(d_g\Omega+W)^{\otimes (n-p)}\right),
\end{align}
where these matrices are indexed by all elements of $C^n$. 
$[W]_{x^iy^i}=0$ for $x=y$ while $\tilde{\Lambda}'$ is only nonzero for elements $[\tilde{\Lambda}']_{xy}$ with $x=y$, so any terms involving 
a Hadamard of $W$ and $\tilde{\Lambda}'$ are 0. Similarly,
the $\Omega$ in 
the $p^{th}$ tensor product is only nonzero for $x^p=y^p$, but for these 
inputs, the term $(\tilde{\Delta}_p^f)'$ is always zero, so in fact the
non-zero terms of this $\Omega$
do not contribute. Thus we are free to replace this $\Omega$ with $d_g\Omega\circ\Delta^g_q$.
We obtain
\begin{align}
0\preceq &d_g^{n-1}\tilde{\Lambda}'\circ\Omega^{\otimes n}\pm(\tilde{V}'\circ(\tilde{\Delta}_p^f)')\circ\left((d_g\Omega+W)^{\otimes (p-1)}\otimes(d_g\Omega\circ\Delta^g_q+W\circ\Delta^g_q)\otimes(d_g\Omega+W)^{\otimes (n-p)}\right)\nonumber\\
0\preceq&d_g^{n-1}\tilde{\Lambda}'\circ\Omega^{\otimes n}\pm(\tilde{V}'\circ(\tilde{\Delta}_p^f)')\circ\left((d_g\Omega+W)^{\otimes n}\circ \{J^{\otimes (p-1)}\otimes\Delta^g_q\otimes J^{\otimes (n-p)}\}\right).
\end{align}
Finally, the term $(\tilde{\Delta}_p^f)'$ can be written as
$J-G$ acting on only 
the $p^{\text{th}}$ term in the tensor product 
$(d_g\Omega+W)^{\otimes n}$, so we need to evaluate $(J-G)\circ (d_g\Omega+W)\circ \Delta_q^g$. We have $(J-G)\circ \Omega=\Delta_q^g\circ
\Omega=0$, and $(J-G)\circ W=W$, so we can remove $(\tilde{\Delta}_p^f)'$
without altering the expression.

Now the term $\{J^{\otimes (p-1)}\otimes\Delta_q^g\otimes J^{\otimes (n-p)}\}$ is almost $\Delta^{f\circ g}_{(p,q)}$, except it is like a primed matrix; its indeces are all elements in $C^{n}$, not just valid inputs to $f$, yet it is not primed, in that some of its elements to non-valid inputs to $f$ are non-zero. However it is involved in a Hadamard product with $\tilde{V}'$, a primed matrix, so all of the terms corresponding to non-valid inputs are zeroed, and we can make it be a primed matrix without affecting the expression. We obtain
\begin{align}
0\preceq&d_g^{n-1}\tilde{\Lambda}'\circ\Omega^{\otimes n}\pm(\tilde{V}'\circ\left((d_g\Omega+W)^{\otimes n}\circ (\Delta^{f\circ g}_{(p,q)})'\right),
\end{align} which is precisely the positivity constraint of the dual program.

Finally, we need to check that $\text{Tr}(cd_g^{n-1}\tilde{\Lambda}'\circ\Omega^{\otimes n})=1$:
\begin{align}
\text{Tr}(cd_g^{n-1}\tilde{\Lambda}'\circ\Omega^{\otimes n})&=cd_g^{n-1}\sum_{a\in S}[\Lambda]_{aa}\sum_{x:\tilde{x}=a}\prod_i[\Omega]_{x^ix^i}\nonumber \\
&=cd_g^{n-1}\sum_{a\in S}[\Lambda]_{aa}\prod_i\sum_{x^i:g(x^i)=a_i}[\Omega]_{x^ix^i} \nonumber \\
&=cd_g^{n-1}\left(\frac{1}{2}\right)^n=1,
\end{align}
where all of the tricks here follow similarly from the discussion following Eq. \eq{obj}.
\end{proof}
\lem{lemone} now follows from \lem{bigidea} along with a simple inductive argument.

\bibliography{wavelet2}

\begin{thebibliography}{10}

\bibitem{Ambainis2000}
Andris Ambainis.
\newblock Quantum lower bounds by quantum arguments.
\newblock In {\em Proc. 32nd ACM STOC}, pages 636--643, 2000.

\bibitem{Ambainis2003}
Andris Ambainis.
\newblock Polynomial degree vs. quantum query complexity.
\newblock {\em J. Comput. Syst. Sci.}, 72(2):220--238, 2006.

\bibitem{Ambainis2010}
Andris Ambainis, Lo\"{i}ck {Magnin}, Martin {Roetteler}, and J\'{e}r\'{e}mie
  {Roland}.
\newblock {Symmetry-assisted adversaries for quantum state generation}.
\newblock In {\em Proc. 24th IEEE CCC}, pages 167--177, 2011.

\bibitem{Childs2007}
Andrew~M. {Childs}, Richard {Cleve}, Stephen~P. {Jordan}, and David {Yeung}.
\newblock {Discrete-query quantum algorithm for NAND trees}.
\newblock {\em Theory of Computing}, 5(1):119--123, 2009.

\bibitem{FarhiNAND1}
Edward Farhi, Jeffrey Goldstone, and Sam Gutmann.
\newblock A quantum algorithm for the hamiltonian nand tree.
\newblock {\em Theory of Computing}, 4(1):169--190, 2008.

\bibitem{Hoyer2007}
Peter H{\o}yer, Troy Lee, and Robert \u{S}palek.
\newblock Negative weights make adversaries stronger.
\newblock In {\em Proc. 39th ACM STOC}, pages 526--535, 2007.

\bibitem{Hoyer2008}
Peter H{\o}yer, Jan Neerbek, and Yaoyun Shi.
\newblock Quantum complexities of ordered searching, sorting, and element
  distinctness.
\newblock {\em Algorithmica}, 34:429--448, 2008.

\bibitem{Randothers2}
Troy Lee, Rajat Mittal, Ben~W. Reichardt, Robert \u{S}palek, and Mario Szegedy.
\newblock Quantum query complexity of state conversion.
\newblock In {\em Proc. 52nd IEEE FOCS}, pages 344 --353, 2011.

\bibitem{Reichardtrefl}
Ben~W. Reichardt.
\newblock Span programs and quantum query complexity: The general adversary
  bound is nearly tight for every boolean function.
\newblock In {\em Proc. 50th IEEE FOCS}, pages 544--551, 2009.

\bibitem{Reichardt2010}
Ben~W. Reichardt.
\newblock Reflections for quantum query algorithms.
\newblock In {\em Proc. 22nd ACM-SIAM SODA}, pages 560--569, 2011.

\bibitem{Reichardtorigin}
Ben~W. Reichardt and Robert \u{S}palek.
\newblock Span-program-based quantum algorithm for evaluating formulas.
\newblock In {\em Proc. 40th ACM STOC}, pages 103--112, 2008.

\bibitem{Saks1986}
Michael Saks and Avi Wigderson.
\newblock Probabilistic boolean decision trees and the complexity of evaluating
  game trees.
\newblock In {\em Proc. 27th IEEE FOCS}, pages 29--38, 1986.

\bibitem{us}
Bohua Zhan, Shelby Kimmel, and Avinatan Hassidim.
\newblock Super-polynomial quantum speed-ups for boolean evaluation trees with
  hidden structure.
\newblock In {\em Proc. 3rd ACM ITCS}, pages 249--265, 2012.

\end{thebibliography}
\bibliographystyle{plain}

\end{document}